\documentclass[aps,pra,twocolumn,nolongbibliography,groupedaddress,10pt]{revtex4-1}
\usepackage{graphicx}
\usepackage{epstopdf}
\usepackage{braket}
\usepackage{lipsum}

\usepackage{amsmath,amssymb,amsfonts,amsthm}
\usepackage{multirow}
\usepackage{verbatim}
\usepackage{url}
\usepackage{comment}
\usepackage{mathtools}
\usepackage{epsf,pgf,graphicx}
\usepackage{color}
\usepackage{tikz,pgf}
\usepackage{makecell}

\newtheorem{theorem}{Theorem}

\newtheorem{example}{Example}

\DeclarePairedDelimiter\ceil{\lceil}{\rceil}

\newcommand\Mycomb[2][^n]{\prescript{#1\mkern-0.5mu}{}C_{#2}}

\usepackage{adjustbox}
\usepackage[colorlinks = true,
            linkcolor = blue,
            urlcolor  = blue,
            citecolor = blue,
            anchorcolor = blue]{hyperref}

\usepackage{subcaption}

\begin{document}
\title{Optimal T depth quantum circuits for implementing arbitrary Boolean functions}
\author{Suman Dutta$^{1}$, Anik Basu Bhaumik$^{2}$, Anupam Chattopadhyay${^2}$, Subhamoy Maitra${^1}$}
\affiliation{$^1$Applied Statistics Unit, Indian Statistical Institute, Kolkata, India, $^2$College of Computing \& Data Science, Nanyang Technological University, Singapore}
\email{sumand.iiserb@gmail.com, anikbasu001@e.ntu.edu.sg\\anupam@ntu.edu.sg, subho@isical.ac.in}
\begin{abstract}
In this paper we present a generic construction to obtain an optimal T depth quantum circuit for any arbitrary $n$-input $m$-output Boolean function $f: \{0,1\}^n \rightarrow \{0,1\}^m$ having algebraic degree $k \leq n$, and it achieves an exact Toffoli (and T) depth of $\lceil \log_2 k \rceil$. This is a broader generalization of the recent result establishing the optimal Toffoli (and consequently T) depth for multi-controlled Toffoli decompositions (Dutta et al., Phys. Rev. A, 2025). We achieve this by inspecting the Algebraic Normal Form (ANF) of a Boolean function. Obtaining a benchmark for the minimum T depth of such circuits are of prime importance for efficient implementation of quantum algorithms by enabling greater parallelism, reducing time complexity, and minimizing circuit latency—making them suitable for near-term quantum devices with limited coherence times. The implications of our results are highlighted explaining the provable lower bounds on S-box and block cipher implementations, for example AES. 
\end{abstract}
\maketitle 

\paragraph*{Introduction--}
\label{sec:intro}
Quantum computing is one of the most fundamental aspects of quantum physics \cite{feynman1986}. Needless to mention, that quantum algorithms often rely on the efficient encoding of classical Boolean functions into quantum circuits through a process known as oracle construction. An oracle is a reversible quantum circuit that implements a Boolean function and is integral to several landmark algorithms, including Deutsch-Jozsa~\cite{dj}, Grover's search~\cite{grover}, Simon's algorithm~\cite{simon}, and Shor's factoring algorithm~\cite{shor}.

Quantum gates, represented by unitary matrices, are the basic building blocks of quantum circuits. Unlike classical gates, they are inherently reversible. In fault-tolerant quantum computing \cite{jones2013pra}, implementing quantum oracles is particularly challenging due to the high resource requirements associated with multi-controlled Toffoli (MCT) gates, which must be decomposed into a universal gate set such as Clifford+T \cite{bravyi2005pra}. This decomposition directly impacts key resource metrics, including gate count, circuit depth, and qubit requirements, each critical for the feasibility of large-scale quantum computations. Hence, minimizing these resources is essential for minimizing errors and improving algorithmic efficiency.

Circuit depth is calculated by the number of layers of gates, assuming parallel execution on disjoint qubits. The T depth explicitly measures the number of such layers containing T gates, a crucial measure in fault-tolerant quantum computation. Since T gates are resource-intensive to implement, often requiring costly magic state distillation, minimizing T depth directly reduces circuit latency and execution time. This is particularly important for near-term quantum devices with limited coherence times. As a result, T depth optimization has become a key focus in the design of efficient quantum circuits \cite{Selinger2013pra,amy2013}.

In \cite[Page 11]{dutta2025pra}, the authors emphasized the importance of estimating resource requirements for arbitrary $n$-input, $m$-output Boolean functions, generalizing beyond multivariate AND functions. Motivated by this, we present, for the first time, optimal T depth quantum circuits for arbitrary Boolean functions. Our construction, which generalizes and subsumes the methodology of \cite{dutta2025pra}, achieves minimal Toffoli (and consequently T) depth and broadly applies to oracle construction and S-box synthesis. While the approach incurs a significant qubit overhead, it offers a compelling trade-off by substantially reducing T depth, thereby enhancing the feasibility of practical quantum algorithm implementations. Our theoretical results are explained in practical terms by analyzing the popular and standard block cipher AES, that has a complicated Boolean circuit.

\paragraph*{Preliminaries--}
\label{sec:pre}
Let $\mathbb{F}_2=\{0,1\}$ be the prime field of characteristic $2$ and $\mathbb{F}^n_2\equiv \{\mathbf{x}=\left(x_1, \ldots, x_n \right):x_i\in\mathbb{F}_2,1\leq i \leq n \}$ be the vector space of dimension $n$ over $\mathbb{F}_2$. An $n$-input $m$-output Boolean function $f$ is defined as a mapping $f:\mathbb{F}^n_2\rightarrow\mathbb{F}^m_2$. The set of all $n$-input $m$-output Boolean functions is denoted by $\mathcal{B}_n^m$. For $m=1$, the set of all $n$-input, single-output Boolean functions is given by $\mathcal{B}_n^1\equiv \mathcal{B}_n$.

The Algebraic Normal Form (ANF) of a Boolean function $f\in\mathcal{B}_n$ can be represented by a polynomial over $\mathbb{F}_2$ in $n$ binary variables, given by
\begin{align*}
    f(x_1, \ldots, x_n) =& a_0 \oplus \bigoplus_{I\subseteq \{1,\ldots , n\}}a_{I}\prod_{i\in I} x_i
\end{align*}
where $a_{I}\in\mathbb{F}_2$. The algebraic degree of a Boolean function $f\in\mathcal{B}_n$, denoted as $\deg (f)$, is given by the maximum cardinality of $I$, such that $a_{I}\neq 0$. Clearly, $\deg(f) \leq n$. A Boolean function $f\in\mathcal{B}_n$ is nonlinear if $\deg(f) > 1$, affine if $\deg(f) = 1$, and constant if $\deg(f) = 0$.

Given $f\in\mathcal{B}^m_n$, the output $f(\mathbf{x})=\mathbf{a}$ does not uniquely determine the input $\mathbf{x}$, as $f$ is inherently irreversible. In contrast, quantum operations are reversible by nature (excluding measurement), and thus the corresponding quantum circuit must be reversible. The quantum circuit implementing $f$ operates on $n+m$ qubits, where the first $n$ qubits encode the input state $\ket{\mathbf{x}}$, and the final $m$ qubits, initialized to $\ket{y}$, stores the functional output, defined as $U_f: \ket{\mathbf{x}}\ket{\mathbf{y}} \rightarrow \ket{\mathbf{x}}\ket{\mathbf{y}\oplus \mathbf{f(\mathbf{x})}}.$

Given the ANF of a Boolean function $f\in\mathcal{B}_n$, the reversible quantum circuit of $f$ can be constructed using CNOT gates for linear terms and multi-controlled Toffoli (MCT) gates for nonlinear terms, with the respective input variable(s) as control qubit(s) and the output qubit as the target. Since MCT gates are non-native to quantum hardware, they must be decomposed into a universal gate set for practical implementation.

The Clifford+T gate set is the most widely used universal gate set in quantum computing due to its compatibility with fault-tolerant implementations. The Clifford group comprises the Hadamard, Phase, and CNOT gates. Augmenting this set with the T gate yields universality, enabling the approximation of any unitary operation to arbitrary precision.

Recently, substantial efforts have been directed toward optimizing the Clifford+T decomposition of the Toffoli gate. State-of-the-art techniques employ measurement-based uncomputation, achieving Toffoli decompositions with four T gates. Notably, Gidney’s construction \cite{gidney2018quantum} attains a T depth $1+1$ without any ancilla qubit (Fig. \ref{fig:logicand}) by executing the initial T gates in parallel. In contrast, Soeken’s design \cite{jaques2020eurocrypt} achieves a T depth of 1 using a single reusable ancilla qubit (Fig. \ref{fig:mathias}). A comprehensive summary of optimized Toffoli decompositions across key metrics is provided in \cite[Table 1]{dutta2025pra}.
\begin{figure}[htbp]
\centering
\begin{subfigure}{.45\textwidth}
    \centering
    \includegraphics[width=1\linewidth]{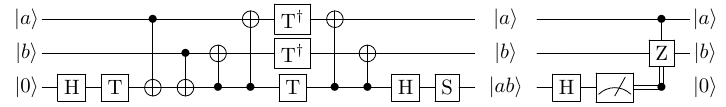}
    \caption{}
    \label{fig:logicand}
\end{subfigure}\\
\begin{subfigure}{.45\textwidth}
    \centering
    \includegraphics[width=1\linewidth]{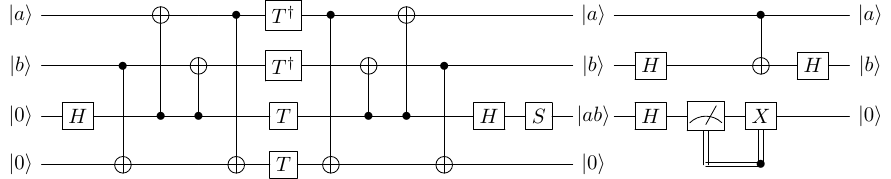}
    \caption{}
    \label{fig:mathias}
\end{subfigure}
\caption{Toffoli decomposition using four T gates, (a) without ancilla using logical-AND \cite{gidney2018quantum}, (b) with a single reusable ancilla \cite{jaques2020eurocrypt}.}
\label{fig:toffoli}
\end{figure}

In addition to optimizing basic Toffoli gates, considerable progress has been made towards the efficient implementations of multi-controlled Toffoli gates. Recent work~\cite{dutta2025pra} demonstrates that employing a binary tree structure in an $n$-MCT circuit decomposition achieves an optimal T depth of $\ceil{\log_2 n}$, utilizing a maximum of $2(n-1)$ ancilla qubits and $4(n-1)$ T gates.
A comprehensive summary of recent advancements in MCT circuit decompositions is provided in \cite[Table II]{dutta2025pra}.

Our theoretical contribution is in the following section, where we present optimal Toffoli (consequently T) depth quantum circuits implementation for evaluating any Boolean function that subsumes the idea of \cite{dutta2025pra} and the fundamental observation is to note the XOR (GF$(2)$ addition) of AND (GF$(2)$ multiplication) in ANF along with the tree implementation. This we will explain in the next section. 
\paragraph*{Optimal T depth quantum resource estimation--}
\label{sec:cont1}
In this section, we present an optimal T depth quantum circuit for implementing any arbitrary $n$-input $m$-output Boolean function $f$, successfully completing the extension idea outlined in \cite[Page 11]{dutta2025pra}. 

As discussed earlier, implementing higher-degree terms in the ANF of a Boolean function $f$ requires multi-controlled Toffoli (MCT) gates, where the corresponding variables act as control qubits and the qubit storing the output serves as the target. Specifically, a degree-$k$ term necessitates a $k$-MCT gate. According to \cite[Corollary~1]{dutta2025pra}, a $k$-MCT gate can be decomposed into the Clifford+T gate set via Toffoli decomposition, yielding a T count of $4(k-1)$ and a T depth of $\lceil \log_2 k \rceil$. Using a binary tree structure, the implementation of the highest-degree term $x_1x_2\ldots x_n$ in an $n$-variable Boolean function incurs a T count of $4(n-1)$ and a T depth of $\lceil \log_2 n \rceil$.

Additionally, the ANF of an arbitrary Boolean function $f \in \mathcal{B}_n$ can contain up to $2^n - (n + 1)$ non-linear terms. More precisely, it may include at most $\Mycomb[n]{k}$ degree-$k$ terms, each requiring a $k$-controlled Toffoli gate in its quantum implementation. If all $k$-MCT gates are executed in parallel, the MCT depth becomes $1$, with the largest gate being an $n$-MCT, having T depth $\lceil \log_2 n \rceil$. This implies that the complete ANF of $f\in\mathcal{B}_n$ can be implemented in quantum with T depth $\ceil{\log_2 n}$.

Needless to mention that, to implement all the MCT gates in parallel, multiple copies of the input variables are required, along with distinct ancilla qubits to store the corresponding outputs. These copies are generated using CNOT gates, which, being Clifford operations, do not contribute to the T depth of the circuit. Summarizing these observations, we present the following theorem.
\begin{theorem}
\label{th:1}
    Let $f\in \mathcal{B}^m_n$ be an $n$-input $m$-output Boolean function. Then, the Clifford+T decomposition of the quantum circuit implementing $f$ can be realized with an optimal T depth of $\ceil{\log_2 n}$, utilizing the following resources:
    (i) at most $2^{n-1}\left(3n-2 \right) - 3n + m +1$ reusable ancilla qubits, (ii) a total of $2^{n+1}\left(n-2 \right) + 4$ T gates, and (iii) a maximum of $2^{n-1}\left(11n + 2m -18 \right) -4n - m +9$ CNOT gates, with (iv) a CNOT depth of $2^n+2n+9\ceil{\log_2 n}-3$.
\end{theorem}
\begin{proof}
There are $\Mycomb[n]{k}$ many degree-$k$ monomials, each of length $k$. Thus, computing all such nonlinear terms in parallel requires $\sum_{k=2}^{n} k\Mycomb[n]{k}$ copies of the input variables. Since one copy of each input already exists, this results in an additional $\big[\sum_{k=2}^{n} k\Mycomb[n]{k} - n\big]$ ancilla qubits, along with an equal number of CNOT gates for both computation and uncomputation. Each contributes a CNOT depth of $\ceil{\log_2\left(\sum_{k=2}^{n}\Mycomb[n-1]{k-1} - n\right)} = n-1$ each. Additionally, $\sum_{k=2}^{n}\Mycomb[n]{k}$ ancilla qubits are required to store the outputs of these MCT gates in parallel.

Once all possible nonlinear monomials in the ANF are implemented, any set of $m$ single-output Boolean functions can be realized by copying at most $\big[\sum_{k=2}^{n}\Mycomb[n]{k} + n\big]$ monomials per function onto $m$ ancilla qubits. This requires at most $m(\sum_{k=2}^{n}\Mycomb[n]{k} + n)$ CNOT gates and a maximum CNOT depth of $(2^n-1)$.

Finally, by \cite[Corollary 1]{dutta2025pra}, the optimal T depth decomposition of each $k$-MCT gate requires $2(k-1)$ ancilla qubits and $9(k-1)$ CNOT gates, with ancilla qubits made reusable via measurement-based uncomputation. Since, one T depth corresponds to a CNOT depth of 9, this contributes an overall CNOT depth of $9\ceil{\log_2 n}$. Hence,
\begin{itemize}
    \item \# Ancilla: $\big[\sum_{k=2}^{n} k\Mycomb[n]{k} - n\big] + \sum_{k=2}^{n}\Mycomb[n]{k} +m + \sum_{k=2}^{n} 2(k-1)\Mycomb[n]{k} = 2^{n-1}\left(3n-2 \right) - 3n +m +1$,
    \item \# T gates: $\sum_{k=2}^{n} 4(k-1)\Mycomb[n]{k} = 2^{n+1}(n-2) +4$,
    \item \# CNOT gates: $2\big[\sum_{k=2}^{n} k\Mycomb[n]{k} - n\big] + m\big[\sum_{k=2}^{n}\Mycomb[n]{k} + n\big] + \sum_{k=2}^{n} 9(k-1)\Mycomb[n]{k} = 2^{n-1}\left(11n +2m- 18 \right)- 4n -m +9$.
    \item CNOT depth: $2(n-1)+(2^n-1) + 9\ceil{\log_2 n}= 2^n+ 2n + 9\ceil{\log_2 n} - 3$.
\end{itemize}
\end{proof}
Theorem \ref{th:1} applies directly to the quantum circuit synthesis of S-boxes, where an $n$-bit S-box corresponds to an $n$-input, $n$-output Boolean function. Accordingly, the resultant circuit requires at most $2^{n-1}\left(3n-2 \right) - 2n+1$ reusable ancilla qubits, $2^{n+1}\left(n-2 \right) + 4$ T gates, and up to $2^{n-1}\left(13n - 18 \right)-5n +9$ CNOT gates. The CNOT depth is bounded by $2^n+2n+9\ceil{\log_2 n}-3$, while achieving an optimal T depth of $\lceil \log_2 n \rceil$. To illustrate the construction, we present the 3-bit S-box used in LowMC \cite{lowmc2015eurocrypt} in the following example.
\begin{example}
    Here, the coordinate Boolean functions are given by: $f_0 = x_0\oplus x_1x_2$, $f_1 = x_0 \oplus x_1 \oplus x_0x_2$, and $f_2 = x_0 \oplus x_1 \oplus x_2\oplus x_0x_1$. Since the maximum algebraic degree among $f_1,f_2,f_3\in\mathcal{B}_3$ is $2$, the corresponding quantum circuit (see Fig. \ref{fig:lowmc-T}) achieves the optimal T depth of $\ceil{\log_2 2} = 1$, using 9 ancilla qubits, 12 T gates, and at most 33 CNOT gates. Notably, as the ANFs of these functions do not contain all possible monomials, the resource requirements are significantly lower than the worst-case estimates in Theorem \ref{th:1}. In practice, such sparsity often leads to substantially reduced overheads compared to theoretical bounds.
    \begin{figure*}[htbp]
        \centering
        \includegraphics[width=0.9\linewidth]{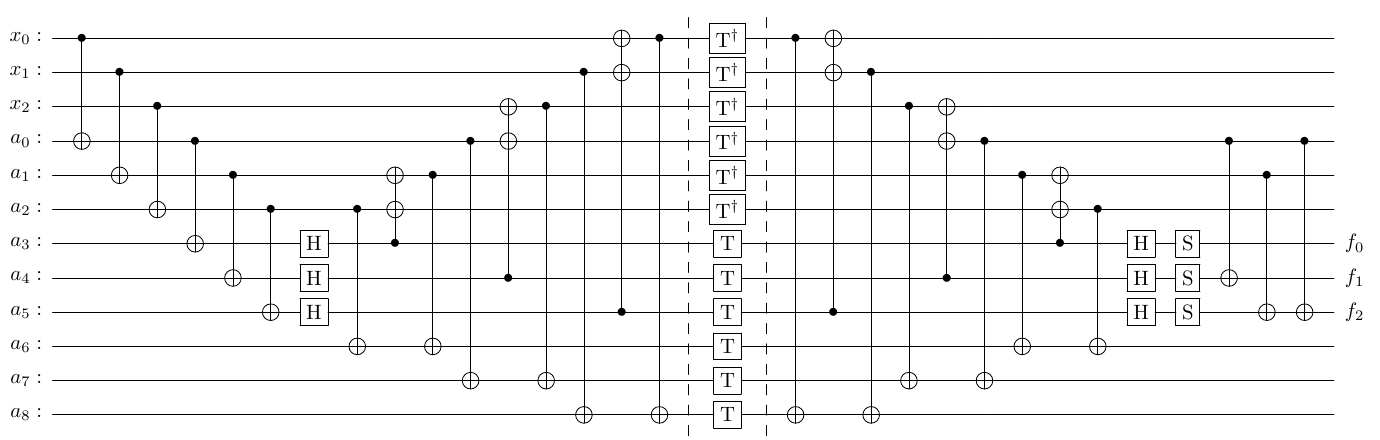}
        \caption{Quantum circuit implementing a 3-bit S-box (used in LowMC) with T depth 1.}
        \label{fig:lowmc-T}
    \end{figure*}
\end{example}

From Theorem \ref{th:1}, the optimal T depth Clifford+T decomposition of an $n$-input, single-output Boolean function $f\in\mathcal{B}_n$ requires at most $2^{n-1}\left(3n-2 \right) - 3n + 2$ reusable ancilla qubits, $2^{n+1}\left(n-2 \right) + 4$ T gates, and up to $2^{n-1}\left(11n -16 \right)-4n +8$ CNOT gates, with a CNOT depth of $2^n + 2n +9\ceil{\log_2 n} -3$. We demonstrate the circuit construction with a representative example.
\begin{example}
    Let $f=x_0x_2 \oplus x_1x_3 \oplus x_0x_1x_2x_3$ be a Boolean function with three nonlinear terms. Since $\deg(f) = 4$, the Clifford+$T$ decomposition achieves a T depth $\lceil \log_2 4 \rceil = 2$. The resulting circuit requires 12 ancilla qubits, 20 T gates, and 46 CNOT gates, with a CNOT depth of 12 (see Fig. \ref{fig:f-T}). Similarly, as $f$ does not include all possible monomials, the resource usage is substantially lower than the theoretical upper bound.
\begin{figure*}[htbp]
    \centering
    \includegraphics[width=0.9\linewidth]{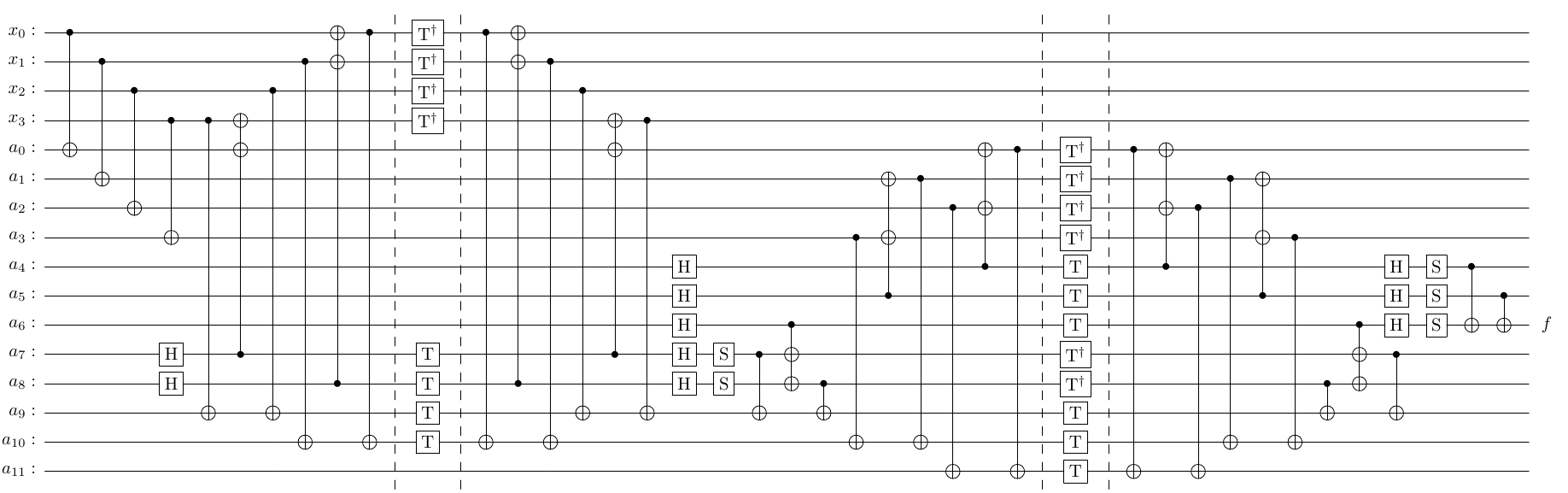}
    \caption{Quantum circuit implementing $f(x_0,x_1,x_2,x_3) = x_0x_2 \oplus x_1x_3 \oplus x_0x_1x_2x_3$ with T depth 2.}
    \label{fig:f-T}
\end{figure*}
\end{example}

While our construction achieves optimal T depth, it incurs exponential overhead in ancilla qubits and CNOT gates with respect to the number of variables. A more practical alternative is to allow a slight increase in T depth in exchange for a substantial reduction in ancilla and CNOT counts. For instance, replacing the T-depth-1 Toffoli decomposition from \cite{jaques2020eurocrypt} (Fig. \ref{fig:mathias}) with the logical-AND-based decomposition from \cite{gidney2018quantum} (Fig. \ref{fig:logicand}) in Theorem \ref{th:1}, reduces the ancilla count by $2^{n-1}\left( n-2\right)+1$ and the CNOT count by $3\cdot 2^{n-1}\left( n-2\right)+3$, while increasing the T depth from $\lceil \log_2 n \rceil$ to $\lceil \log_2 n \rceil + 1$. This trade-off highlights a practical direction for future research, wherein our construction can serve as a benchmark for T depth, while optimizing other quantum resources, even at the cost of slight T depth sub-optimality, rather than prioritizing depth reduction alone.

That is, this observation explains the optimal T depth of ANF-based implementations. Particularly, in circuits designed for cryptanalytic purposes, this optimal depth, as a benchmark, has not been explored earlier in a disciplined manner as a performance metric. In the next section, we illustrate this gap with concrete examples, using the widely studied AES block cipher. The idea will follow for any standard block ciphers in general.

\paragraph*{Optimal T depth quantum circuit of AES--}
\label{sec:cont2}
This section presents an optimal T depth quantum circuit for the AES algorithm. We first construct an optimal T depth implementation of the AES S-box. Then, we extended this to the full AES algorithm for key sizes of 128, 192, and 256 bits, corresponding to 10, 12, and 14 rounds, respectively.

AES is a symmetric-key block cipher standardized by NIST, operating on 128-bit blocks. The internal state is represented as a $4\times 4$ matrix $S\in\mathbb{F}^{4\times 4}_{2^8}$, where each element $S_{i,j} \in \mathbb{F}_{2^8}$ corresponds to one byte.

AES can be abstracted as a Boolean function $f:\mathbb{F}_2^{m+k} \rightarrow \mathbb{F}_2^{c}$, where $m$ is the message length (128 bits), $k$ the key length (128, 192, or 256 bits), and $c$ the ciphertext length (128 bits). From Theorem \ref{th:1}, the quantum circuit of AES can be designed with an optimal T depth of $\ceil{\log_2 256} = 8$, $\ceil{\log_2 320} = 9$, and $\ceil{\log_2 384} = 9$, depending on the key length. However, treating AES as a single combinational Boolean circuit is impractical as well as elusive due to its iterative structure. That is, while the T depth of 8 or 9 for the complete implementation of AES is theoretically understood, this cannot be seen as a practical benchmark.

For all practical purposes, AES is implemented round-wise. Consequently, one can design the quantum circuit of AES for each rounds and estimate its T depth based on actual implementation. Each AES round (except the last one) consists of four transformations, described as follows.
\begin{itemize}
    \item \textbf{SubBytes}: Applies the AES S-box to each of $S_{i,j} \in \mathbb{F}_{2^8}$ in the internal state. This is the only nonlinear operation and the primary contributor to T depth, requiring 16 parallel S-box evaluations: $S_{i,j}\leftarrow \text{S-box}(S_{i,j})$.
    \item \textbf{ShiftRows}: Performs a cyclic left shift of the rows. For row $i$, the transformation is given by $S_{i,j}\leftarrow S_{i,(j+i)\mod 4}$. Since, this is a SWAP operation, while designing the quantum circuit, this is implemented via rewiring and does not require any additional quantum resources.
    \item \textbf{MixColumns}: Applies a linear transformation to each column $\mathbf{c}\in\mathbb{F}_{2^8}^{4}$ of the internal state $S$ using an MDS matrix $M\in\mathbb{F}_{2^8}^{4\times 4}$: $\mathbf{c}\leftarrow M\cdot\mathbf{c}$. Since, this is a linear transformation, it can be implemented using only CNOT gates. Note that the final round of the AES omits the MixColumns operation.
    \item \textbf{AddRoundKey}: The AES key schedule algorithm expands the primary key $K$ to generate round keys $K_0, K_1,\ldots , K_r\in\mathbb{F}_{2^8}^{4\times 4}$, each used in a specific round. In this transformation, the internal state $S$ is XORed with a round key $K_i$, defined as $S\leftarrow S \oplus K_i$.
\end{itemize}

Since only the SubBytes transformation contains nonlinear components (S-box), specifically, MCT gates contributing to T depth, the overall T depth of an AES implementation is dominated by the S-box operations. As 16 S-boxes are evaluated in parallel per round, the total T depth of the AES circuit can be estimated as $r$ times the T depth of a single S-box, where $r$ denotes the number of rounds. We now present a detailed quantum resource estimation of the AES S-box with optimal T depth, as derived from Theorem~\ref{th:1}.

The coordinate Boolean functions of the AES S-box have a maximum algebraic degree of 7. Following Theorem \ref{th:1}, the quantum circuit of AES S-box can thus be designed with an optimal T depth of $\ceil{\log_2 7} = 3$.

The ANF of all coordinate functions combined includes all possible monomials except the degree 8 term $x_0x_1\ldots x_7$. Constructing these monomials in parallel requires $\big[\sum_{k=2}^{7} k\cdot\Mycomb[8]{k} - 8\big]= 1000$ ancilla qubits to copy the input variables, and an additional $\sum_{k=2}^{7}\Mycomb[8]{k} = 246$ ancilla qubits to store the monomials. The functional output of eight coordinate functions requires 8 more ancilla qubits. The implementation of $\sum_{k=2}^{7}\Mycomb[8]{k}$ MCT gates requires a further $\sum_{k=2}^{7}2(k-1)\Mycomb[8]{k} = 1524$ ancilla qubits. This leads to a total cost of $(1000+ 246+ 8+ 1524) = 2778$ ancilla qubits.

Copying the input bits for parallel construction requires $\big[\sum_{k=2}^{7} k\cdot\Mycomb[8]{k} - 8 \big]= 1000$ CNOT gates, with a CNOT depth of $\ceil{\log_2\left( \sum_{k=2}^{7}\Mycomb[7]{k-1}\right)} = 7$, applicable to both computation and uncomputation. Across all eight coordinate functions, there are 1001 monomials in total, requiring 1001 CNOT gates. The CNOT depth here is determined by the coordinate function with the largest number of monomials, which is 145. The MCT gate implementation incurs $\sum_{k=2}^{7}9(k-1)\Mycomb[8]{k} = 6858$ CNOT gates. Thus, the total CNOT count is $(2\times 1000)+1001+6858 = 9859$. As each Toffoli gate corresponds to a CNOT depth of 9 (see Fig. \ref{fig:mathias}), the total CNOT depth is given by $(2\times 7)+145+(3\times 9)=186$. The total number of T gates required is $\sum_{k=2}^{7}4(k-1)\Mycomb[8]{k} = 3048$.

Using Gidney’s logical-AND decomposition from Fig.~\ref{fig:logicand} instead of the T-depth-1 approach from Fig.~\ref{fig:mathias} reduces the ancilla count to 2016 and the CNOT count to 7573, at the expense of increasing the T depth to $\ceil{\log_2 7}+1 = 4$. Here, each Toffoli gate has a CNOT depth of 6, resulting in a total CNOT depth of $(2\times 7) + 145 + (3\times 6) = 177$.
\begin{table}[htbp]
    \centering
    \caption{Quantum circuits for AES S-box with corresponding resource estimates.}
    \resizebox{0.48\textwidth}{!}{
    \begin{tabular}{c c c c c c}
    \hline
    \hline
        Toffoli-to-T & Ancilla & \#CNOT & CNOT depth & \#T & T depth\\[0.1cm]
        \hline
        Using Fig. \ref{fig:logicand} & 2016 & 7573 & 177 & 3048 & 4\\[0.1cm]
        Using Fig. \ref{fig:mathias}  & 2778 & 9859 & 186 & 3048 & 3\\
        \hline
        \hline
    \end{tabular}}
    \label{tab:AES-Sbox}
\end{table}

In the SubBytes transformation, 16 AES S-boxes are executed in parallel, implying that the transformation requires 16 times more ancilla qubits, CNOT gates, and T gates compared to a single S-box, while the depths remain unchanged. The ShiftRows transformation incurs no additional quantum resource overhead. Moreover, \cite[Table 5]{yuan2024ToSC} shows that the MixColumns transformation can be implemented in place (without additional ancilla) using 98 CNOT gates with a CNOT depth of 13. Since, the MixColumns transformation is applied in parallel to four columns of the internal state, the total CNOT count becomes $4\times 98 = 392$, while the CNOT depth remains 13. Finally, bit-wise XOR of the round key involves 128 parallel CNOT gates, incurring CNOT depth 1. The overall quantum resource requirements for a single AES round are summarized in Table~\ref{tab:AES-1round}.
\begin{table}[htbp]
    \centering
    \caption{Quantum implementation of a single round of AES with corresponding resource estimates.}
    \resizebox{0.48\textwidth}{!}{
    \begin{tabular}{c c c c c c}
    \hline
    \hline
        Toffoli-to-T & Ancilla & \#CNOT & CNOT depth & \#T & T depth\\[0.1cm]
        \hline
        Using Fig. \ref{fig:logicand} & 32256 & 121688 & 191 & 48768 & 4\\[0.1cm]
        Using Fig. \ref{fig:mathias} & 44448 & 158264 & 200 & 48768 & 3\\
        \hline
        \hline
    \end{tabular}}
    \label{tab:AES-1round}
\end{table}

Assume that AES runs for $r\in\{10,12,14\}$ rounds, depending on the key length. Ancilla qubits used in the S-boxes of each round are reclaimed for subsequent rounds through measurement-based uncomputation, except for the 8 qubits required to store the functional output. Consequently, the ancilla count increases by $8\times 16 = 128$ per round, yielding a total ancilla requirement of $44448+128(r-1)$. Furthermore, since the final round omits the MixColumns operation, the total CNOT count is given by $158264r -392$, and the CNOT depth is $200r-13$. Both the T count and T depth scale linearly with the number of rounds. 

\begin{table}[htbp]
    \centering
    \caption{Optimal T depth quantum implementation of complete AES with corresponding resource estimates.}
    \resizebox{0.48\textwidth}{!}{
    \begin{tabular}{c c c c c c c}
    \hline
    \hline
        Key & Round & Ancilla & \#CNOT & CNOT depth & T count & T depth\\[0.1cm]
        \hline
        128 & 10 & 45600 & 1582248 & 1987 & 487680 & 30\\[0.1cm]
        192 & 12 & 45856 & 1898776 & 2387 & 585216 & 36\\[0.1cm]
        256 & 14 & 46112 & 2215304 & 2787 & 682752 & 42\\
        \hline
        \hline
    \end{tabular}}
    \label{tab:AES}
\end{table}

Table~\ref{tab:AES} presents the resource estimates for complete quantum implementations of AES with different key sizes, using the optimal T depth construction (Fig. \ref{fig:mathias}).
The T depth achieved in this construction is optimal, as no quantum circuit for AES (when the rounds are implemented one after another) can attain a lower T depth, regardless of the number of ancilla qubits or other resource overheads.

Table~\ref{tab:AES} highlights the practicality of our optimal T depth quantum circuit constructions for block ciphers. While earlier implementations have reported fewer ancilla qubits or lower gate counts, and in some cases, comparable T depths (e.g., \cite[Table 7]{aes2022asiacrypt} reports a T depth of 60 for AES and AES$^\dagger$ combined (that means 30 fpr each), and\cite[Table 13]{aes2023asiacrypt} reports T depths of 40, 48, and 56 for AES-128, 192, and 256, respectively), none of them formally establish the optimality of their constructions. Given that AES is one of the most important ciphers in symmetric-key cryptography, numerous brute-force attempts have been made to minimize various implementation parameters. The observed low T depths appear to be the outcome of such heuristic efforts rather than a systematic construction that proves the optimality across a large class of block ciphers. To the best of our knowledge, this is the first instance to formally report that T depths of 30, 46, and 42 are optimal for quantum implementations of AES-128, AES-192, and AES-256, respectively, and to demonstrate that such optimal T depth circuits can be realized in practice. Furthermore, instead of relying on brute-force optimization, we propose a generic construction that yields optimal T depth quantum circuits for a broad class of block ciphers, thereby setting a definitive benchmark for T depth in quantum cryptographic implementations. Table \ref{tab:AES-summary} summarizes the earlier implementations of AES towards T depth reduction along with the other resource estimates.
\begin{table*}[htbp]
    \centering
    \caption{Optimal T depth quantum implementation of complete AES: A comparison with previous works.}
    \begin{tabular}{c c c c c c c}
    \hline
    \hline
        Key size & References & Ancilla & ~CNOT count~ & ~CNOT depth~ & ~T count~ & ~T depth~\\[0.1cm]
        \hline
        128 & \cite[Table 4]{jaques2020eurocrypt} & 4244 &  284420 & NA & 54400 & 120\\[0.1cm]
        128 & \cite[Table 13]{aes2023asiacrypt} & 3689 & 132376 & NA & 27200 & 40\\[0.1cm]
        128 & \cite[Table 7]{aes2022asiacrypt} & 5576 & 285393 & NA & 62400 & 30\\[0.1cm]
        128 & [Present work] & 37464 & 1441128 & 1987 & 487680 & 30\\[0.1cm]
        \hline
        192 & \cite[Table 4]{jaques2020eurocrypt} & 4564 &  321021 & NA & 60928 & 144\\[0.1cm]
        192 & \cite[Table 13]{aes2023asiacrypt} & 3945 & 149256 & NA & 30464 & 48\\[0.1cm]
        192 & [Present work] & 37480 & 1729824 & 2387 & 585216 & 36\\[0.1cm]
        \hline
        256 & \cite[Table 4]{jaques2020eurocrypt} & 4884 &  393534 & NA & 75072 & 168\\[0.1cm]
        256 & \cite[Table 13]{aes2023asiacrypt} & 4457 & 187128 & NA & 38080 & 56\\[0.1cm]
        256 & [Present work] & 37496 & 2017736 & 2787 & 682752 & 42\\[0.1cm]
        \hline
        \hline
    \end{tabular}
    \label{tab:AES-summary}
\end{table*}

\paragraph*{Conclusions--}
\label{sec:con}
Given the ANF of an arbitrary $n$-input, $m$-output Boolean function $f$ having algebraic degree $k$, this work presents the first construction of an optimal T depth quantum circuit for $f$, along with a complete resource estimation. The optimization criteria here is only related to the minimum T depth that has definite importance in terms of circuit implementation for quantum paradigm. Thus, the other required resources are high. We recommend that during any circuit design, first it must be understood what should be the benchmark T depth of the circuit and then only other parameters may be optimized, as minimization of T depth is one of the most important criteria in quantum circuit design. That is, the benchmark for optimal T depth for Boolean circuits implemented in quantum algorithms are completely settled with this treatment and we present the design of block ciphers as concrete examples.



\end{document}